\documentclass[11pt]{article}
\usepackage{fullpage}
\usepackage[numbers]{natbib}
\usepackage[utf8]{inputenc}
\usepackage{amsmath,amsthm,amssymb}
\usepackage{thmtools} 
\usepackage{thm-restate}
\usepackage{enumitem}
\usepackage[normalem]{ulem}
\usepackage{xspace}
\usepackage[usenames,dvipsnames]{xcolor}
\usepackage[colorlinks=true,citecolor=ForestGreen,linkcolor=ForestGreen,urlcolor=NavyBlue]{hyperref}

\usepackage[ruled]{algorithm2e} 

\SetAlFnt{\small}
\SetAlCapFnt{\small}
\SetAlCapNameFnt{\small}
\SetAlCapHSkip{0pt}
\IncMargin{-\parindent}

\usepackage{natbib}

\setcitestyle{acmnumeric}

\usepackage{amsfonts,dsfont}
\usepackage{graphicx}
\allowdisplaybreaks
\usepackage[suppress]{color-edits}
\addauthor{sm}{blue}
\addauthor{rr}{purple}
\addauthor{dm}{red}

\usepackage{microtype}
\usepackage{multirow}
\usepackage{colortbl}
\usepackage{setspace}
\usepackage{nicefrac}
\graphicspath{{images/}}
\usepackage{blindtext}

\newtheorem{theorem}{Theorem}
\newtheorem{lemma}{Lemma}
\newtheorem{definition}{Definition}

\newtheorem{corollary}{Corollary}
\newtheorem{proposition}{Proposition}

\usepackage[capitalise,nameinlink]{cleveref}

\newcommand{\E}{\mathds{E}}
\renewcommand{\Pr}{\mathds{P}}
\newcommand{\ind}{\mathds{1}}
\newcommand{\sigs}{\mathbf s}

\title{Optimal Stopping with Interdependent Values\thanks{The work of S.\ Mauras, and D.\ Mohan has been partially funded by the European Research Council (ERC) under the European Union's Horizon 2020 research and innovation program (grant agreement No. 866132), by an Amazon Research Award, by the NSF-BSF (grant No. 2020788), and by a grant from TAU Center for AI and Data Science (TAD). This research was partially conducted while Mauras and Mohan were visiting the Simons Laufer Mathematical Sciences Institute (formerly MSRI) in Berkeley, California, during the Fall 2023 semester, which was funded by the National Science Foundation (grant No. DMS-1928930) and by the Alfred P. Sloan Foundation (grant G-2021-16778).}}
\author{Simon Mauras
            \thanks{INRIA, Tel Aviv University; {\tt simon.mauras@inria.fr}}
            \and
            Divyarthi Mohan
            \thanks{Tel Aviv University; {\tt divyarthim@tau.ac.il}}
            \and
            Rebecca Reiffenhäuser
            \thanks{University of Amsterdam; {\tt r.e.m.reiffenhauser@uva.nl}}
            }

\begin{document}
\begin{titlepage}
\maketitle
\begin{abstract}
We study online selection problems in both the prophet and secretary settings, when arriving agents have interdependent values. In the interdependent values model, introduced in the seminal work of Milgrom and Weber [1982], each agent has a private signal and the value of an agent is a function of the signals held by all agents. Results in online selection crucially rely on some degree of independence of values, which is conceptually at odds with the interdependent values model.
For prophet and secretary models under the standard independent values assumption, prior works provide constant factor approximations to the welfare. On the other hand, when agents have interdependent values, prior works in Economics and Computer Science provide truthful mechanisms that obtain optimal and approximately optimal welfare under certain assumptions on the valuation functions.

We bring together these two important lines of work and provide the first constant factor approximations for prophet and secretary problems with interdependent values. We consider both the algorithmic setting, where agents are non-strategic (but have interdependent values), and the mechanism design setting with strategic agents. All our results are constructive and use simple stopping rules.
\end{abstract}

\end{titlepage}

\section{Introduction}
Consider a single-item auction, say for a piece of art, where buyers arrive online. The goal is to sell the item to the agent with the highest value, while making the decisions about whether to select each buyer immediately on their arrival. Additionally, buyers' values can depend on one another: a buyer interested in decorating their living room might be influenced by the impression of those arriving before him, and spontaneously attribute a higher value to the item if it is very popular.
A buyer that sees the item as a pure investment, on the other hand, will be interested in its resale value alone, which is fully determined only after the arrival of the very last buyer. We analyze settings of the above type, formally, of online selection processes with interdependent values. This means, we combine concepts from the areas of online selection/optimal stopping with those from the theory on interdependent values, both of which have raised a lot of recent interest due to their central and important applications in economics.

\paragraph{Online selection.} 
In the online settings we consider, formally, a sequence of $n$ numbers arriving in an online fashion. The goal is to select the highest number, with the restriction that for each one, we have to irrevocably select or reject it at the time of arrival. Without additional assumptions, no online algorithm can achieve nontrivial competitive ratio in this setting.

The two central models considered to enable close-to-optimal competitive ratios w.r.t. the (expected) maximum are prophet inequalities, and the secretary model.
In prophet inequalities, originally introduced in the $70$s, each of the $n$ numbers is assumed to be drawn independently from a known distribution $D_i$ on arrival. This assumption allows for a (best-possible) $2-$approximation to the expected optimum via simple threshold policies (e.g., \citet{KrengelS78}).\
In the secretary setting, the impossibility of obtaining good online algorithms is instead circumvented by the assumption that all numbers will arrive not adversarially, but in \emph{uniformly random order}.
The famous, original secretary problem dates back even further, where an optimal $\tfrac 1e$-approximation to the maximum (originally for the ordinal variant) was known since the early 60s (see \citet{Dynkin63}).
Both models have in past years fueled a large variety of research directions, comprising many combinatorial problem variants and constraints on the selectable sets of online elements. Most relevant to us, there are strong connections to economics (especially online auctions and e-commerce) since applications comprise a large number of auction settings.

\paragraph{Interdependence.} The celebrated interdependent values (IDV) model, introduced by \citet{MilgromWeber82} building up on \citet{wilson1969communications},\footnote{The 2020 Economics Nobel Prize was awarded to Milgrom and Wilson for their work on IDV and auction design~\cite{nobel2021considerations}.} is well-studied in the economics literature when considering settings where agents have partial information and their values may depend on the information of all bidders. For instance, suppose we have a house for sale, different potential buyers might have different partial information about the house (e.g., one may have information regarding the school district of the neighborhood and another might have a better assessment of the structural integrity) and a buyer's value for the house can be influenced by any or all of these information. In the IDV model, for a single-item allocation problem, each agent $i$ has a private signal $s_i$ and a public valuation function $v_i(\cdot)$ that maps the signals of all buyers $(s_1,\ldots,s_n)$ to a value for the item; that is, $v_i(\sigs)$ is $i$'s value for the item given a signal profile $\sigs=(s_1,\ldots,s_n)$. There is a long literature in economics and computer science studying mechanism design in the interdependent values model. While in the standard model with \emph{private values} the well-known VCG auction obtains optimal welfare truthfully, in the IDV model truthful welfare maximization is possible if and only if the valuation functions satisfy the \emph{single-crossing condition}~\cite{maskin1992,DM00,ausubel1999generalized}. Informally, under the single-crossing condition, each agent's signal has the most impact on their own valuation function compared to others' valuation functions. Recent work in EconCS takes an algorithmic approach and investigates approximation guarantees. Of particular interest to this paper is the result that when the valuation functions satisfy \emph{submodularity over signals} (SOS), there are truthful mechanisms that obtain a constant factor approximation to the optimal welfare~\cite{EdenFFGK19,LuSZ22,EdenFGMM23}. Informally, a valuation function is submodular over signals if for each $j$ the impact of increasing $s_j$ is higher when the other signals $s_{-j}$ are lower. 

\paragraph{Online selection with interdependence.} In order to marry the two directions above and obtain online algorithms for settings with interdependent valuations, we assume that agents arrive online, and the algorithm has to select one of them to sell an item or service to. We make the standard assumption 
that agents' valuation functions are public and known beforehand. However, agents have private signals (e.g., opinions they form on some specific property of a house when seeing it) which are drawn from independent distributions\footnote{Note that the assumption of underlying distributions is not necessary in the secretary model, but is w.l.o.g. as long as we make no further assumptions on the distributions themselves.}.

For our results on interdependent values in online selection, an important distinction is whether or not the agents' values can depend on future (and hence, yet undetermined) signals. 
We therefore make different modeling assumptions on the nature of the underlying online market, specifically on the set of signals that can influence each agent's value:
\begin{itemize}
\item In a first setting, we consider \emph{myopic agents}, who promptly consume the item once selected; and hence receive a value dependent only on signals observed so far, minus potentially a \emph{prompt} payment charged immediately by the algorithm (which as well can only depend on signals observed so far). 
\item In a second setting, we consider \emph{farsighted agents}, who if selected will benefit from the item also in the future; and hence receive a value dependent on all signals, minus potentially a \emph{tardy} payment charged only at the end of the online algorithm (which can depend on all signals).
\end{itemize}
Note that we do not give results for myopic agents with tardy payments (since for the former, payments can also be determined promptly), or for farsighted agents with prompt payments (since these are only implementable when assuming that agents do not have any knowledge of other agents' signals or even distributions beforehand, i.e. not for the commonly used notions of incentive compatibility).

\subsection{Our results} We consider agents with interdependent values for a single item, in both the prophet and secretary models. In contrast to previous work, we do not restrict our view to submodular-over-signals valuations, but only make the weaker assumption of subadditivity (which, notably, constitutes a natural boundary for constant-factor online algorithms, which are incompatible with the existence of arbitrary complements).
Our results comprise the algorithmic setting, and the strategic setting with selfish agents for which we give (ex-post) incentive compatible mechanisms. We summarize them in the table below. In addition to the results listed in the table, note that for the stronger assumption of submodular over signals valuations, the factor of $2e$ in the secretary algorithm with farsighted/myopic agents can be improved to a $4$, see \cref{app:submodular}. 

\begin{table}[h]
\centering
\begin{tabular}{|c|c|c|c|}
\cline{2-4}
\multicolumn{1}{c|}{}& Agents & Algorithm & Mechanism\\
\hline
\multirow{2}{*}{Prophet} 
& farsighted & $\Omega(n)$ (\Cref{thm:prophet_farsighted_hard}) & $\Omega(n)$\\
\cline{2-4}
& myopic & $\leq 4$ (\cref{thm:prophet_myopic_algo}) & $\leq 8$ (\cref{thm:prophet_myopic_mechanism})\\
\hline
\multirow{2}{*}{Secretary} 
&farsighted & $\leq 2e$ (\Cref{thm:secretary_farsighted_algo}) & $\leq 4e$ (\Cref{thm:secretary_farsighted_mechanism})\\
\cline{2-4}
& myopic  & $\leq 2e$ (\Cref{thm:secretary_farsighted_algo}) &$\leq 4e$ (\Cref{thm:secretary_farsighted_mechanism})\\
\hline
\end{tabular}
\caption{Approximation ratio of our algorithms (no incentive constraints) and mechanisms (EPIC), when agents have subadditive-over-signals valuations.}
\end{table}
In particular, we obtain constant-factor approximation algorithms and mechanisms -- or prove their impossibility -- for each possible combination. This constitutes not only the first, but a close-to complete picture of the extent to which interdependent values and central paradigms in online selection can be combined. Notably, our algorithms exactly lose a factor of $2$ compared to the standard independent values settings, or a factor $4$ when we also consider incentives. Moreover, when removing the dependence of the agents' valuations on any signal except their own, our algorithms recover the original tight ratios from the prophet inequality and secretary problem.

Our constants are in general not optimal; 
this is in part due to the fact that for the corresponding offline settings with interdependence, optimal constants are also yet unknown.

\subsection{Related Work}

\paragraph{Interdependent values.} There is a rich economics literature on interdependent settings over the past 50 years. A common impossibility result emerging in the literature states that truthful welfare maximization is only possible if the valuations satisfy a strong condition such as single-crossing~\cite{ausubel1999generalized,DM00,JehieM01,JehielMMZ06,ItoP06,CKK15}. In recent years the computer science literature has seen much interest in studying the IDV model through the lens of approximation, in order to circumvent these impossibilities (e.g.,~\cite{RoughgardenTC16, ChawlaFK14, EdenFFG18, EdenFFGK19, EdenGZ22,CohenFMT23,gkatzelis2021prior,ChenEW,EdenFTZ21}). \citet{EdenFFG18} consider a setting with an approximate single-crossing condition and obtain approximately optimal welfare truthfully. 
In a breakthrough result, \citet{EdenFFGK19} establish a $4$-approximation when the valuations satisfy submodularity\footnote{\citet{EdenFFGK19} observe in their conclusion that some of their results extend to subadditive valuation functions.} over signals without any single-crossing type assumption and also extend the results to combinatorial auctions under an additional separability condition. \citet{LuSZ22} provide an improved  approximation bound of $3.315$ for single-item auctions and \citet{AmerTC21} provide a $2$-approximation in the special case of binary signals. More recently, works of \citet{EdenGZ22,EdenFGMM23} consider the more general setting where the valuations are private and establish a constant approximation under submodular valuations for single-item and multi-unit auctions . 

\citet{RoughgardenTC16} and \citet{Li13} study simple prior-independent mechanisms that obtain approximately optimal revenue under different assumptions, and \citet{ChawlaFK14} further minimize the assumptions needed. All these works assume some form of single-crossing type condition on the valuations. 

Prior works have also considered interdependent values in other settings beyond auctions. For example, \citet{chakraborty2010two} study interdependence in matching markets, \citet{CohenFMT23} study the public projects setting with interdependent values and \citet{BirmpasELR23} consider interdependence in the fair division problem.

\paragraph{Prophet inequality.}
Prophet Inequalities, originally introduced by \citet{KrengelS77,KrengelS78} and \citet{SamuelCahn84}, are one of the most central concepts in decision making for stochastic settings. After being employed for algorithmic mechanism design in online markets by \citet{HajiaghayiKS07} and \citet{ChawlaSH10}, prophet inequalities have been obtained for a large variety of prominent problem settings, e.g. online selection with matroid constraints \cite{KleinbergW12}, or online matchings \cite{AlaeiHL12}, with a strong focus on economic settings like combinatorial auctions \cite{DuttingFKL20,CorreaC23}. 
While most results crucially exploit independence of the value distributions $D_i$, few results are also known for restricted types of dependence (see e.g. \cite{SamuelCahn91,ImmorlicaSW23}).
Our work, in a similar spirit, follows a new approach to incorporating dependence between online values: here, while the distributions of online signals are independent, valuations are obtained from signals in a dependent fashion. A model of similar type (for myopic agents), to the best of our knowledge, has only been captured previously by \cite{brunel1979parier}, who prove the existence of a $2(1+\sqrt{3})$-approximation when the valuations are subadditive over signals, with a non-constructive proof.

\paragraph{Secretary problem.}
Following the optimal solution  for the original secretary problem by \citet{Dynkin63}, a rich body of work has introduced extensions and applications in various directions, perhaps most famous among which are matroid secretary problems initiated by \citet{BabaioffIK07}. Similarly to prophet inequalities, one major focus has been on mechanism design. For example, the optimal approximation factor of $\tfrac{1}{e}$ has been recovered for bipartite matchings \cite{Reiffenhauser19} and XOS combinatorial auctions \cite{AbelsKRV13}. Our results contribute to the large body of work on the above paradigms, by extending the range of prophet and secretary algorithms to applications with interdependent valuations. 
\section{Preliminaries and Model}

We consider a problem where $n$ agents have interdependent values: each agent $i\in [n]$ holds some private signal $s_i \in \mathbb R_+$ and a publicly known monotone valuation function $v_i: \mathbb R_{\geq 0}^n \rightarrow \mathbb R_{\geq 0}$, where monotonicity means that $v_i$ is non-decreasing w.r.t. the input vector. We denote $\mathbf s = (s_1, \dots, s_n)$, and for every subset $X\subseteq [n]$ we write $\mathbf s_X = (\ind_{1\in X}\cdot s_1, \dots, \ind_{n\in X}\cdot s_n)$, that is, we replace $s_i$ by $0$ if $i\notin X$.

We focus on online settings, where agent $t\in [n]$ arrives at time $t$ and we observe the signal $s_t$ (while, as is standard in the interdependent values literature, the valuation functions are publicly known). 
We say that an agent $i$ is:
\begin{itemize}
\item \emph{myopic}, if their value only depends on the signals received so far, i.e. is equal to $v_i(\sigs_{[i]})$.
\item \emph{farsighted}, if their value depends on all signals, i.e. is equal to $v_i(\sigs_{[n]})$.
\end{itemize}
Myopic agents model situations where the value is obtained instantly on selection, e.g. when agents bid for a good and their value depends on the signals they observed from agents arriving previously. Farsighted agents model the setting where e.g. a good is assigned right now, but the opinion of people that have not yet arrived influences the value it has to the winning agent (for instance, an investment they might want to re-sell later).
We design online algorithms, which observe agents one at a time, deciding whether to continue (and reject the current agent) or to stop (and select, i.e. assign the good to the current agent). The objective is to maximize the expected \emph{social welfare}, that is, the expected value of the agent selected by the algorithm.

For approximation purposes, we compare our algorithms to the expectation of the maximum value (in hindsight). In the secretary models, where instead of distributions, an instance is just a set of $n$ fixed signal values, the benchmark accordingly reduces to be the maximum such value. Note that this benchmark depends on whether agents are farsighted or myopic. The two settings are not directly comparable, as both the social welfare of the algorithm and of the benchmark are larger with farsighted agents.

\subsection{Subadditive valuation functions}

We consider a natural class of valuations called subadditive over signals (or simply subadditive). It captures contexts where signals (information) are not complements, roughly referring to the notion that they do not increase in value by combining them. This includes most settings explored in the literature, such as the mineral rights model~\cite{wilson1969communications} and the resale model~\cite{Klemperer98,myerson1981optimal}.

\begin{definition}[Subadditive over signals]
We say a valuation function $v(\cdot)$ is \emph{subadditive} over signals, if for any signal profile $\sigs$ and any $X\subseteq [n]$ we have
\[
v(\sigs)\le v(\sigs_X) + v(\sigs_{[n]\setminus X})
\]
\end{definition}

We note that the class of subadditive valuation functions is  strictly more general than the class of valuation functions that are \emph{submodular over signals}, which are well-studied in the IDV literature~\cite{EdenFFGK19,AmerTC21,LuSZ22,EdenGZ22,CohenFMT23,EdenFGMM23}.

\subsection{Incentive compatibility}
 In order to incentivize the agents to report truthfully, mechanisms (usually) charge the winning agent some payment. In general, a mechanism is an allocation algorithm together with a payment rule. When agents are farsighted, we allow those payments to be tardy, i.e. to only specify the paid amount after the last agent has arrived. When agents are myopic, we require the algorithm to have prompt payments, i.e. the price is specified and paid immediately when an agent is chosen.

For each agent $i$ let $x_i(\sigs_{[i]})$ be the indicator variable where $x_i = 1$ if and only if $i$ is selected by the given algorithm and $x_i = 0$ else. Let $p_i(\sigs) \ge 0$ denote the payment charged (wlog $p_i = 0$ if $x_i = 0$). Notice that in a prompt mechanism, $p_i$ only depends on reported signals $\sigs_{[i]}$.

\begin{definition}[EPIC]
   A mechanism $(\mathbf x, \mathbf p)$ is \emph{Ex-Post Incentive Compatible (EPIC)} if {truth-telling is a Nash equilibrium, that is if} for every $i\in [n], \sigs, s'_i$ we have
   \[x_i(\sigs_{[i]})\cdot v - p_i(\sigs) \ge x_i(\sigs_{[i-1]},s'_i)\cdot v - p_i(\sigs_{-i},s'_i) \]
   where $v$ is the value of agent $i$, which is equal to $v_i(\sigs)$ if agent $i$ is farsighted, and equal to $v_i(\sigs_{[i]})$ if agent $i$ is myopic.
\end{definition}

\citet{RoughgardenTC16} give a sufficient (and necessary) condition for an allocation rule to be implementable truthfully, which we recall in \Cref{lem:truthful}.

\begin{lemma}\label{lem:truthful}
For any deterministic allocation rule $x_i(\sigs_{[i]})$ which is monotone in $s_i$, that is
$$
\forall i,\forall \sigs, \forall s_i' \geq s_i,\qquad
x_i(\sigs_{[i]}) \leq x_i(\sigs_{[i-1]}, s_i'),
$$
there exists a payment $p_i(\sigs)$ such that the mechanism $(\mathbf x,\mathbf p)$ is EPIC. Moreover, if agents are myopic, then the price $p_i$ only depends on signals $\sigs_{[i]}$.
\end{lemma}
\begin{proof}
From \citet{RoughgardenTC16}. The payment is equal to
$$
p_i(\sigs) = x_i(\sigs_{[i]}) \cdot \inf\{v_i(\sigs_{-i}, s_i')\,|\, s_i'\geq 0\text{ such that }x_i(\sigs_{[i-1]}, s_i') = 1\}
$$
when agent $i$ is farsighted; and is equal to
$$
p_i(\sigs_{[i]}) = x_i(\sigs_{[i]}) \cdot \inf\{v_i(\sigs_{[i-1]}, s_i')\,|\, s_i'\geq 0\text{ such that }x_i(\sigs_{[i-1]}, s_i') = 1\}
$$
when agent $i$ is myopic. That is, the payment of the winning agent $i$ equals the minimum value of $i$ such that $i$ remains the winner, given the signals of the other agents.
\end{proof}
Although we make only very limited use of randomization in our results, note that  (universally) truthful, randomized mechanisms can be defined as lotteries over deterministic EPIC mechanisms.

\section{The Prophet Model}
We first consider the prophet model, where agents arrive online in an adversarial order and the signals are drawn independently from a known prior distribution. Formally, we have $n$ agents, each characterized by a signal distribution $D_i$ and a valuation function $v_i(\cdot)$. An adversary dictates the order in which agents arrive; without loss of generality, we relabel the agents so that agent $t$ arrives at time $t$. Upon the arrival of agent $i$, their signal $s_i$ is independently drawn from the distribution $D_i$. At this point, we need to decide whether to irrevocably reject the agent and continue the selection process, or to accept the agent and conclude the selection. Our goal is to design simple online algorithms that maximize the expected value of the accepted agent. We evaluate the performance of our algorithms by comparing it to the expected maximum value.

In \Cref{sec:prohpet_farsighted} we study the setting with far-sighted agents and  show a strong impossibility that no online algorithm can obtain better than $\Omega(n)$-approximation to the expected maximum value. In \Cref{sec:prophet_myopic_algo} we consider myopic agents and provide a simple online algorithm that obtains a $4$-approximation to the expected maximum welfare. Finally, we extend this result to settings with incentive constraints in \Cref{sec:prophet_myopic_mech}.


\subsection{Impossibility with far-sighted agents}
\label{sec:prohpet_farsighted}
We show that for farsighted agents, no algorithm can guarantee a competitive ratio below $\Omega(n)$. Our proof, on a high level, captures the fact that when values are determined by the very last signal to arrive, any choice made before that is reduced to essentially guessing.

\begin{theorem}\label{thm:prophet_farsighted_hard}
In the prophet setting with farsighted agents, any algorithm has a competitive ratio of at least $\Omega(n)$.
\end{theorem}
\begin{proof}
Assume that signal $s_n$ is drawn uniformly in $[0,1]$ and that each agent $i$ has a value $v_i(\mathbf s) = 2^i \cdot \ind[s_n \geq 1-1/2^i]$. Then the expected maximum value is
$$
\E[OPT] = \sum_{i=1}^{n-1} 2^i\cdot \Pr[1-1/2^{i+1} > s_n \geq 1-1/2^i] + 2^n\cdot \Pr[s_n1 \geq 1-1/2^n] = \frac{n+1}{2}
$$
Note that in the myopic setting, all agents always have valuation $0$, except for the $n\textsuperscript{th}$ agent. For farsighted agents, however, the realized signal $s_n$ (drawn from distribution $D_n$)
determines the point in time at which the sequence of increasing $(2^i)$-values stops (and only zeroes arrive from then on). 
Any deterministic algorithm stops at a fixed $i$ and yields an expected value of exactly $1$ (for any choice of $i$). This is because $v_i(\sigs) =2^i$ with probability $1/2^i$ and $0$ otherwise.
Using Yao's Lemma~\cite{Yao77}, randomized algorithms cannot give any improved approximation ratio on this random instance. Thus, we cannot do better than an $\Omega(n)$ approximation.
\end{proof}

Notice that an equivalent example, replacing $s_n$ with $s_1$, shows that the algorithm must indeed observe signals, and not only the agents' values on the signals of agents arrived so far.

\subsection{Algorithm with myopic agents}
\label{sec:prophet_myopic_algo}

Given the impossibility of any constant approximation with farsighted agents, we switch our focus to myopic agents. However, the following property observes that without any assumption (such as subadditivity) on the complementarity of signals, it is not possible to guarantee a competitive ratio below $\Omega(n)$.

\begin{proposition}\label{prop:prophet_general_hardness}
In the prophet setting with myopic agents and general valuation functions, any algorithm has a competitive ratio of at least $\Omega(n)$.
\end{proposition}
\begin{proof}
Assume that all signals $s_i$ are draw i.i.d. and uniformly from $\{0, 2\}$, and that agent $i$ has a value equal to $v_i(\sigs_{[i]}) = \prod_{j=1}^i s_j$. This exactly reproduces the construction of \Cref{thm:prophet_farsighted_hard}, where the value of an agent double at each step, until some unpredictable time when it drops to zero for all remaining agents. Using the same argument as in the previous section, no algorithm can guarantee better than a $\Omega(n)$ approximation.
\end{proof}

A model which addresses the impossibilities raised by \Cref{thm:prophet_farsighted_hard,prop:prophet_general_hardness} was previously considered in the work of  \citet{brunel1979parier}, who showed (non-constructively) that there exists a $2(1+\sqrt{3}) \approx 5.46$ approximation when agents have subadditive valuations that do not depend on future signals. In this section, we provide a simple stopping rules that guarantee an improved ratio of $4$.

While our algorithm's main idea is based on the classic threshold approach from prophet inequalities (where the threshold is set to half the expected optimum value), dealing with interdependent values even in this simplest setting we consider is not without challenges. In particular, opposed to classic prophet inequalities, the simple threshold policy does not work.
Consider the following instance: all agents except for agent $n$ have valuation equal to $s_1+1$. Agent one's distribution is such that $s_1=0$ with probability $(1-\varepsilon)$, and $s_1=\tfrac {1}{\varepsilon}$ with very small probability $\varepsilon$.
Agent $n$ has valuation equal to $n\cdot s_1$.
The expected maximum is therefore $n$, half of which will be set as the algorithm's threshold value.
Consider now the rare case that indeed, $s_1$ has nonzero value (due to $s_1=\tfrac{1}{\varepsilon}$ being drawn). Every agent beats the threshold of $\tfrac{n}{2}$, and the algorithm will realize only a value of $\tfrac{1}{\varepsilon}$ while the optimum is by a factor $n$ higher.
Given that for the above distribution, the standard threshold does not obtain constant approximation, we might attempt to use a higher threshold instead. However, consider instead the case that $D_1$ always returns zero. Now, any constant-approximative threshold algorithm must choose a threshold $\leq 1$.
We therefore need to take care of such phenomena, caused by the interdependence of values, in our algorithms.

We now present a simple stopping rule which achieves a $4$-approximation, only losing a factor $2$ compared to the standard setting with independent values. The main idea to resolve issues caused by interdependence is  that, at any point of the algorithm, we can skip the current agent if her value is surely smaller than any one of the future agents, given the signals we observed. {In particular, the algorithm stops at time $t$ if and only if the value of agent $t$ (i.e., $v_t(\sigs_{[t])}$) satisfies the following two conditions: (1) it is at least as much as the threshold $X = \E[OPT]/2$, and (2) it is no worse than the current lower bound on the value of the future agents $i>t$ (i.e., $v_i(s_{[t]})$).}

\begin{algorithm}
\quad Set threshold
$$X = \E[\max_{i\in[n]} v_i(\sigs_{[i]})]/2.$$
\,\,\,Stop at the first time $t$ such that:
\begin{itemize}
\item $v_t(\sigs_{[t]}) \geq X$, and 
\item $v_t(\sigs_{[t]}) \geq v_i(\sigs_{[t]})$ for all $i > t$
\end{itemize}
\caption{4-approximation algorithm with myopic agents.}
\label{algo:prophet_myopic_algo}
\end{algorithm}

\begin{theorem}\label{thm:prophet_myopic_algo}
\Cref{algo:prophet_myopic_algo} is $4$-competitive. That is, it obtains an expected (myopic) welfare of at least
$\frac{1}{4}\E[\max_{i\in [n]} v_i(\sigs_{[i]})]$.
\end{theorem}

\begin{proof}
    Let $I\in \text{argmax}_i v_i(\sigs{[i]})$ be the random variable that denotes the index of the maximum value agent and let $T\in \{1,\ldots,n,\infty\}$ be the stopping time of \Cref{algo:prophet_myopic_algo}. 
    Notice that either the algorithm stops at $T\le I$ or the algorithm does not stop, i.e. $T=\infty$. This is because the algorithm does not stop by time $I$ only if $v_I(\sigs_{[I]}) < X$, which implies all $v_i(\sigs_{[i]}) < X$. By using the subadditivity of $v_I$, we therefore have for any fixed arrival order and fixed realization of the $n$ signals
    \begin{align*}
v_I(\sigs_{[I]}) &\leq \ind[T = \infty] \cdot v_I(\sigs_{[I]})
\\
&+\ind[T < \infty] \cdot v_I(\sigs_{[T]})\\
&+ \ind[T < I]\cdot v_I(\sigs_{[T+1,I]}).
\end{align*}
Recall that, if the algorithm did not stop then $v_I(\sigs_{[I]})<X$, thus bounding the first term as
\begin{equation}\label{eq:algo1-doesn't-stop}
    \ind[T = \infty] \cdot v_I(\sigs_{[I]}) < \ind[T = \infty]\cdot X.
\end{equation}
However, if the algorithm stops then $v_T(\sigs_{[T]}) \ge v_I(\sigs_{[T]})$, thus bounding the second term as
\begin{equation}\label{eq:algo1-stops-easy}
    \ind[T < \infty] \cdot v_I(\sigs_{[T]}) \le \ind[T < \infty] \cdot v_T(\sigs_{[T]})
\end{equation}
Finally, to deal with the last term we bound it in expectation as follows
\begin{align}\label{eq:algo1-stops}
\E[\ind[T < I]\cdot v_I(\sigs_{[T+1,I]})]
&\leq 
\sum_{t=1}^n \E[\ind[T=t]\cdot \max_{j>t} v_j(\sigs_{[t+1,j]})]\notag\\
& = \sum_{t=1}^n \E[\ind[T=t]]\cdot \E[\max_{j>t} v_j(\sigs_{[t+1,j]})] \notag\\
&\le \Pr[T<\infty]\cdot 2X \notag\\
&\le \Pr[T<\infty]\cdot X + \E[\ind[T<\infty]\cdot v_T(\sigs_{[T]})],
\end{align}
where the first inequality follows by {the law of total probability} and using the fact that $I > t$ to upper-bound $v_I(\sigs_{[t+1,I]})$, the equality observes that stopping at time $t$ is independent of all signals after $t$, the next inequality simply uses the definition of $X$, and finally we obtain the last inequality by observing that $v_T(\sigs_{[T]}) \ge X$ whenever the algorithm stops.

Overall by putting together \Cref{eq:algo1-doesn't-stop,eq:algo1-stops-easy,eq:algo1-stops} we have the following bound on $\E[v_I(\sigs_{[I]})]$
\begin{align*}
\E[v_I(\sigs_{[I]})] &\leq \Pr[T = \infty] \cdot X +\Pr[T < \infty] \cdot X\\
&~+2\cdot \E[\ind[T < \infty] \cdot v_T(\sigs_{[T]})] \\
&\leq X + 2\cdot \E[\ind[T < \infty] \cdot v_T(\sigs_{[T]})].
\end{align*}

Observe that the expected welfare of the algorithm is exactly
$\E[ALG] = \E[\ind[T<\infty]\cdot v_T(\sigs_{[T]})]$
and by definition of $X$ it holds
\[
\E[v_I(\sigs_{[I]})] \leq \frac{1}{2}\cdot\E[v_I(\sigs_{[I]})] +2\cdot \E[ALG],
\]
which simplifies to
\[
\frac{1}{2}\cdot\E[v_I(\sigs_{[I]})]\leq 2\cdot \E[ALG],
\]
yielding a $4$-approximation.
\end{proof}
We remark that the above analysis is indeed crucially fueled by the idea of evaluating all (even future) agents' values on the current set of signals in every step, and only stopping if the current agent is the so-far maximum (including those who haven't arrived).
This is what allows us to relate the value of the maximum agent to that selected by the algorithm. Moreover, it nicely illustrates the importance of public valuation functions: in case they are private, the above counterexample cannot be circumvented since there is no way to identify the presence of the \emph{better} agent before arrival. 

\subsection{Mechanism with myopic agents}
\label{sec:prophet_myopic_mech}

In the previous section, we presented a simple $4$-approximation in the algorithmic setting, i.e. without considering the agents' incentives.
Next, we show how to build a truthful stopping rule (monotone in each agent's signal) which achieves an $8$-approximation, losing an extra factor of $2$ compared to the non-strategic setting.  

Recall that in \Cref{algo:prophet_myopic_algo}, even if the value of agent $t$ exceeded the threshold $X$, we did not stop if a future agent was obviously (i.e. on the currently known set of signals) better. However, since the future agents are evaluated (among others) on the signal $s_t$, agent $t$ may have an incentive to misreport her signal. Therefore, the translation into an incentive-compatible mechanism is no longer immediate (as it is for pure threshold strategies). \dmedit{It is well-known that an allocation rule can be truthfully implemented if and only if each agent $i$'s allocation is monotone non-decreasing in her signal $s_i$ (see \Cref{lem:truthful}). Hence, under the single crossing assumption,\footnote{The valuations satisfy the single crossing conditions if for all $i,j$, $s_i$, $s_{-i}$, we have $\partial_i v_i(\sigs)/\partial s_i \ge \partial_i v_j(\sigs)/\partial s_i$. } we have a truthful mechanism that is a $4$-approximation by charging appropriate payments with \Cref{algo:prophet_myopic_algo}.

\begin{corollary}\label{cor:prophet_SC}
If the valuation functions satisfy the single crossing condition, then \Cref{algo:prophet_myopic_algo} provides an EPIC mechanism by charging price $p_t$ for the selected agent $t$, where
\[
p_t = \max\{ X, \inf\{v_t(s_{[t-1]},s'_t) | s'_t\ge 0 \text{ s.t. } v_t(s_{[t-1]},s'_t) \ge v_i(s_{[t-1]},s'_t) \text{ for all } i>t \}\}.
\]
\end{corollary}

However, without the single crossing assumption, we cannot obtain an EPIC mechanism using \Cref{algo:prophet_myopic_algo}.} To overcome this, we use the power of randomization. In particular, if an agent $t$'s value exceeds the threshold $X=\E[OPT]/2$, then with probability $1/2$ we accept her (and stop), and with probability $1/2$ we decide to accept a future agent with highest estimated value using the signals observed so far, $\sigs_{[t]}$. 

\begin{algorithm}
\quad Set threshold
$$X := \E[\max_{i\in [n]} v_i(\sigs_{[i]})]/2.$$
\,\,\,Let $T$ be the first time $t$ such that $v_t(\sigs_{[t]}) \geq X$.
\begin{itemize}
\item With probability $1/2$, stop at time $T$ (and charge agent $T$ a price of $X$).
\item With probability $1/2$, wait and stop at time $\text{argmax}_{i>t} v_i(\sigs_{[t]})$.
\end{itemize}
\caption{8-approximation mechanism with myopic agents.}
\label{algo:prophet_myopic_mechanism}
\end{algorithm}

\begin{theorem}\label{thm:prophet_myopic_mechanism}
\Cref{algo:prophet_myopic_mechanism} is a 8-approximation, that is, it obtains an expected (myopic) welfare of at least
$\frac{1}{8}\E[\max_{i\in[n]} v_i(\sigs_{[i]})]$.
\end{theorem}
\begin{proof}
The proof follows the same principle as that for the algorithmic setting. Let $I \in \text{argmax}_{i}v_i(\sigs_{[i]})$ be a random variable (index of a maximum agent value), and let $T \in \{1, \dots, n, \infty\}$ be the first time $t$ such that $v_t(\sigs_{[t]}) \geq X$. Note that $T$ is not the stopping time as the algorithm stops only later with probability $1/2$.
By construction, it holds that either $T \leq I$, or $T = \infty$, because if $v_I(\sigs_{[I]}) < X$ then $v_t(\sigs_{[t]}) < X$ for all $t$. Thus, using the subadditivity of $v_I$, we can write
\begin{align*}
v_I(\sigs_{[I]}) &\leq \ind[T = \infty] \cdot v_I(\sigs_{[I]})
\\
&+\ind[T < \infty] \cdot v_I(\sigs_{[T]})\\
&+ \ind[T < I]\cdot v_I(\sigs_{[T+1,I]}).
\end{align*}
If the algorithm did not stop, then $v_I(\sigs_{[I]}) < X$, which bounds the first term as
\[
\ind[T = \infty] \cdot v_I(\sigs_{[I]}) \le \ind[T = \infty] \cdot X
\]
We next bound the last term in expectation,
\begin{align*}
\E[\ind[T < I]\cdot v_I(\sigs_{[T+1,I]})]
&\leq 
\sum_{t=1}^n \E[\ind[T=t]\cdot \max_{j>t} v_j(\sigs_{[t+1,j]})]\\
& = \sum_{t=1}^n \E[\ind[T=t]]\cdot \E[\max_{j>t} v_j(\sigs_{[t+1,j]})] \\
&\le \Pr[T<\infty]\cdot 2X\\
&\le \Pr[T<\infty]\cdot X + \E[\ind[T<\infty] \cdot v_T(\sigs_{[T]})],
\end{align*}
where the first inequality follows by the law of total probability and using that $I>t$ given $T=t$, the equality follows by the observation that stopping at time $t$ does not depend on future signals after $t$, the next inequality simply uses the definition of $X$ to upper-bound some expected values, and the last inequality follows because whenever the $T<\infty$ we have $v_T(\sigs_{[T]}) \ge X$ by definition.

Overall, this gives the following bound on the expected optimum $\E[v_I(\sigs_{[I]})]$
\begin{align*}
\E[v_I(\sigs_{[I]})] &\leq \Pr[T = \infty] \cdot X + \Pr[T < \infty] \cdot X\\
&+\E[\ind[T < \infty] \cdot v_I(\sigs_{[T]})]\\
&+ \E[\ind[T<\infty]\cdot v_T(\sigs_{[T]})].
\end{align*}
Now, observe that by design of the algorithm we have
\begin{equation}\label{eq:mech-prophets}
\E[\ind[T<\infty] \cdot v_T(\sigs_{[T]})] + \E[\ind[T < I] \cdot v_I(\sigs_{[T]})] \leq 2\cdot \E[ALG],
\end{equation}
where we simply bound $v_I(\sigs_{[T]}) \le  \max_{j > T} v_j(\sigs_{[T]})$ for $T < I$.

Thus, plugging in \cref{eq:mech-prophets} and using the definition of $X = \E[v_I(\sigs_{[I]})]/2$ we get
\begin{align*}
    \frac{1}{2}\cdot\E[v_I(\sigs_{[I]})]
   &\leq 2\cdot\E[ALG] +  \E[\ind [T= I] \cdot v_I(\sigs_{[I]})].
\end{align*}
Finally, we bound $\E[\ind [T= I] \cdot v_I(\sigs_{[I]})]$ by $2\cdot \E[ALG]$ to obtain
$$
\frac{1}{2}\cdot\E[v_I(\sigs_{[I]})] \le  4\cdot \E[ALG],
$$
resulting in an $8$-approximation.
\end{proof}

We next show that \Cref{algo:prophet_myopic_mechanism} is indeed a truthful mechanism.

\begin{lemma}\label{lem:prophet-truthful}
    \Cref{algo:prophet_myopic_mechanism} is an EPIC mechanism.
\end{lemma}
\begin{proof}
    We first observe that our mechanism is essentially a posted price mechanism with fixed price $X$, except with probability half we don't sell (and give it away for free to a future agent).  Note that, if an agent $i$ wins for free due to the random coin toss, this uses no information about $s_i$. If an agent $t$ wins because her value exceeds $X$, then she has no reason to misreport because decreasing her signal $s_t$ can only potentially make her lose and increasing the signal doesn't affect her allocation or price. Crucially, we use $s_t$ to estimate the future agents only after we decided to reject $t$ due to the random coin toss.
\end{proof}

\section{The Secretary Model}
We next consider the secretary model, where agents arrive in a random order. More formally, the valuation functions and signals are formed adversarially, then agents are shuffled uniformly at random and relabelled such that agent $t$ arrives at time $t$. In our model, the algorithm does not have access to the valuation functions of agents who have not arrived, but at any given time it can query the values of all agents so far on any subset of observed signals. Our goal is to design simple stopping rules, which guarantee a constant fraction of the maximum value, in expectation over the random order.

\dmedit{In fact, we prove a much stronger statement where even when considering only the myopic welfare of the online algorithm and mechanism (i.e., the value of the accepted agent is evaluated only on the signals observed so far), we prove a constant approximation to the farsighted benchmark (i.e., the maximum value when considering all signals). Since the valuation functions are monotone over the signals, we observe that any algorithm that obtains an $\alpha$-fraction of the farsighted benchmark also obtains an $\alpha$-fraction of the myopic benchmark. Moreover, the farsighted welfare of the algorithm is greater than or equal to the myopic welfare of the algorithm. 

In Section~\ref{sec:secretary_algo} we provide a simple stopping rule such that the (myopic) welfare is a $2e$-approximation to the maximum farsighted value in expectation. In Section~\ref{sec:secretary_mech} we extend our results to obtain a truthful mechanism such that the (myopic) welfare is a $4e$-approximation to the maximum farsighted value in expectation. These results imply a $2e$-approximate algorithm (without incentive constraints) and a $4e$-approximate truthful mechanism for both the myopic and farsighted settings. 

Before we present our stopping rules we discuss two  properties of interest.}
One crucial property of secretary algorithms is that the probability of stopping is independent of the subset of agents who have arrived (but depends on the ordering of that subset). This is formalized in \Cref{lem:secretary_stopping}, which will be used in the analysis of \Cref{algo:secretary_farsighted_algo,algo:secretary_farsighted_mechanism}.
\begin{lemma}\label{lem:secretary_stopping}
In the secretary setting, a fixed set of agents $A$ arrive in a random order $a_1, \dots, a_n$. Consider an algorithm which stops at time $t$ if and only if
\begin{itemize}
    \item $t > k$, for some fixed integer $k\geq 0$, and
    \item $a_t = \texttt{best}(\{a_1, \dots, a_t\})$, where $\texttt{best}$ maps each subset $S\subseteq A$ to an agent $a\in S$.
\end{itemize}
Then the algorithm stops at time $T=t$ with probability equal to $\frac{k}{t(t-1)}$, and this is independent of the (random) set of agents who arrived on or before time $t$. More formally,
$$
\forall t > k,
\forall S\subseteq A\text{ such that }|S|=t,\qquad
\Pr[T = t\,|\,\{a_1, \dots, a_t\} = S] = \frac{k}{t(t-1)}.
$$
\end{lemma}
\begin{proof}
We will show by induction on $t\geq k$ that
$$
\forall S\subseteq A\text{ such that }|S|=t,\qquad
\Pr[T > t\,|\,\{a_1, \dots, a_t\} = S] = \frac{k}{t}.
$$
This is trivially true at time $t=k$, because the stopping rule skips the first $k$ agents. Now, Let us compute the probability that $T > t+1$. For all subset $S\subseteq A$ of size $|S| = t+1$, we have that
$$
\Pr[T > t+1\,|\,\{a_1, \dots, a_{t+1}\} = S] =
\Pr[T > t\text{ and }a_{t+1}\neq\texttt{best}(S)\,|\,\{a_1, \dots, a_{t+1}\} = S].
$$
Using the law of total probability, we pick $a_{t+1}\in S$, and we obtain
$$
\Pr[T > t+1\,|\,\{a_1, \dots, a_{t+1}\} = S] = \frac{1}{|S|}\sum_{\substack{a\in S\\a\neq \texttt{best}(S)}} \Pr[T > t\,|\, \{a_1, \dots, a_t\} = S\setminus \{a\}]
$$
Finally, using the induction hypothesis, we have
$$
\Pr[T > t+1\,|\,\{a_1, \dots, a_{t+1}\} = S] = \frac{|S|-1}{|S|}\cdot \frac{k}{t} = \frac{k}{t+1},
$$
which concludes the induction. Next, for all $t > k$ and for all subset $S\subseteq A$ of size $|S| = t$ we write
\begin{align*}
\Pr[T = t\,|\,\{a_1, \dots, a_t\} = S] &= 
\Pr[T > t-1\text{ and }a_t = \texttt{best}(S)\,|\,\{a_1, \dots, a_t\} = S]\\
&= \frac{1}{|S|}\cdot \Pr[T > t-1\,|\,\{a_1, \dots, a_{t-1}\} = S\setminus\texttt{best}(S)],
\end{align*}
because with probability $1/|S|$ the best of $S$ arrives at time $t$.

We know that the probability that $T > t-1$ is $\frac{k}{t-1}$ from the above argument, and hence we conclude that $T=t$ with probability $\frac{k}{t(t-1)}$.
\end{proof}

A second important property of secretary settings is that given a stopping rule which asymptotically achieves a constant approximation when the number of agents becomes large, one can turn it into a stopping rule which achieves that exact constant for every value of $n$. Indeed, the stopping rule can pretend to observe many dummy agents with value $0$, which will never be selected but are only here to artificially increase the value of $n$.

\subsection{Algorithm with myopic or farsighted agents}\label{sec:secretary_algo}

Equipped with \Cref{lem:secretary_stopping}, we propose a simple stopping rule which achieves a $2e$-approximation, only losing a factor $2$ compared to the standard secretary setting. The main intuition is that after the sampling phase (first $\lfloor n/e\rfloor$ steps), the algorithm knows sufficiently many signals to have a good estimate of the agents true value (on all signals). 

\begin{algorithm}
\quad At step $t$, when agent $t$ arrives, stop if:
\begin{itemize}
\item $t > n/e$ (i.e., skip a constant fraction of agents), and
\item $v_t(\sigs_{[t]}) > v_i(\sigs_{[t]})$ for all $i<t$.
\end{itemize}
\caption{2e-approximation algorithm for the secretary model.}
\label{algo:secretary_farsighted_algo}
\end{algorithm}

\newcommand{\arrived}[1]{A_{#1}}
\begin{theorem}\label{thm:secretary_farsighted_algo}
\Cref{algo:secretary_farsighted_algo} is a $2e$-approximation. {That is, the expected (\dmedit{myopic}) value of the accepted agent is at least $\frac{1}{2e}\max_i v_i(\sigs)$.}
\end{theorem}
\begin{proof}
We define the random variable $T\in \{1, \dots, n, \infty\}$ to be the stopping time of the algorithm. In the secretary setting, $n$ agents from a set $A$ arrive in a uniformly random order $a_1, \dots, a_n$. Recall that we labeled agents according to their arrival order, that is, in the algorithm, 
$$
\forall i\in [n], \forall J\subseteq [n],\qquad
v_i(\sigs_J) := \bar v_{a_i}(\bar \sigs_{\{a_j | j\in J\}}),
$$
where $\bar v$ and $\bar \sigs$ are the original, fixed valuation functions and signals (determined adversarially) before applying the random ordering. In particular, there exists an agent $a^\star\in A$ with the largest value $OPT = \bar v_{a^\star}(\bar\sigs)$. For convenience, we define the set function
$$
\forall X\subseteq A, \qquad f(X) := \bar v_{a^\star}(\bar s_X),
$$
that is, $f(X)$ denotes the estimated value of $a^\star$ only using the signals of $X\subseteq A$.

Next, we define the (random) set $\arrived{t} := \{a_1, \dots, a_t\}$ of agents who have arrived at time $t$. Observe that the stopping rule of \Cref{algo:secretary_farsighted_algo} corresponds to \Cref{lem:secretary_stopping} with $k = \lfloor n/e\rfloor$ and 
$$
\forall S\subseteq A,\qquad \texttt{best}(S) := \text{argmax}_{a\in S} \bar v_a(\bar \sigs_S).
$$
Using \cref{lem:secretary_stopping}, the event where the algorithm stops at time $T=t$ is independent of $\arrived{t}$, and has probability equal to
\begin{equation}\label{eq:secretary_farsighted_algo}
\forall t > n/e,\qquad
\Pr[T = t\,|\,\arrived{t}] = \frac{\lfloor n/e\rfloor}{t(t-1)}
\end{equation}
We write the expected welfare obtained by the algorithm as
\begin{align*}
    \E[ALG] &\ge \sum_{t=\lceil n/e\rceil}^n \E[\ind[T=t]\cdot v_t(\sigs_{[t]})] &\text{\dmedit{(equality holds for myopic)}}\\
    &\geq \sum_{t=\lceil n/e\rceil}^n \E[\ind[T=t]\cdot v_t(\sigs_{[t]})\cdot\ind[a^\star \in \arrived{t}]]&\text{(always smaller)}\\
    &\geq \sum_{t=\lceil n/e\rceil}^n \E[\ind[T=t]\cdot f(\arrived{t})\cdot\ind[a^\star \in \arrived{t}]]&\text{(by the stopping condition)}\\
    &= \sum_{t=\lceil n/e\rceil}^n \frac{\lfloor n/e\rfloor}{t(t-1)}\cdot\E[f(\arrived{t})\cdot\ind[a^\star \in \arrived{t}]]&\text{(using \Cref{eq:secretary_farsighted_algo})}\\
    &= \sum_{t=\lceil n/e\rceil}^n \frac{\lfloor n/e\rfloor}{t(t-1)}\cdot\frac{t}{n}\cdot \E[f(\arrived{t})\,|\,a^\star \in \arrived{t}]&\text{($a^\star\in\arrived{t}$ with probability $t/n$)}
\end{align*}
Next we define
$$\forall t \in [n], \qquad \alpha_t := \E[f(\arrived{t})\,|\,a^\star \in \arrived{t}].$$
which gives the inequality
$$
\E[ALG] \geq \frac{\lfloor n/e\rfloor}{n} \sum_{t=\lceil n/e\rceil}^n \frac{\alpha_t}{t-1}.
$$
Alternatively, $\alpha_t$ it is the expected value of $f(X\cup\{a^\star\})$, given a random subset $X\subseteq A\setminus\{a^\star\}$ of size $|X| = t-1$. In particular, $\alpha_n = OPT$, the optimal farsighted welfare, and the sequence $\alpha_t$ is non-decreasing. By linearity of expectation, and using the fact that $f$ is a monotone subadditive set function, we have that 
\begin{align*}
\alpha_n = f(A) &= \E_{A_t}[f(A)\,|\, a^\star\in \arrived{t}]&\text{(by definition)}\\
&\leq \E_{A_t}[f(A_t) + f(A\setminus A_t)\,|\,a^\star\in \arrived{t}] &(\text{by subadditivity})\\
&\leq \E_{A_t}[f(A_t) + f(A\setminus A_t\cup\{a^\star\})\,|\,a^\star\in \arrived{t}]
&\text{(by monotonicity)}\\&
= \E_{A_t}[f(A_t)\,|\,a^\star\in \arrived{t}]
 + \E_{A_{n-t+1}}[f(\arrived{n-t+1})\,|\,a^\star\in \arrived{n-t+1}]
&\text{(by symmetry)}\\
&= \alpha_t + \alpha_{n-t+1}&\text{(by definition)}
\end{align*}
Finally, we will split the sum from $t_0 = \lceil n/e\rceil$ to $n$ in three: from $t_0$ to $t_1 = \lceil n/(e-1)\rceil$, from $t_1$ to $t_2 = n-\lceil n/e\rceil + 1$, and from $t_2$ to $n$.
By monotonicity of the $\alpha_t$'s, we have that
\begin{align*}
\sum_{t=\lceil n/e\rceil} \frac{\alpha_t}{t-1}& \geq 
\alpha_{t_0} \sum_{t=t_0}^{t_1} \frac{1}{t-1} +  
\alpha_{t_1} \sum_{t=t_1+1}^{t_2-1} \frac{1}{t-1} +  
\alpha_{t_2} \sum_{t=t_2}^{n} \frac{1}{t-1}\\
&\geq 
\alpha_{t_0} \ln\left(\frac{t_1}{t_0-1}\right) +  
\alpha_{t_1} \ln\left(\frac{t_2-1}{t_1}\right) +  
\alpha_{t_2} \ln\left(\frac{n}{t_2-1}\right)\\
&\geq 
\alpha_{t_0} \ln\left(\frac{e}{e-1}\right) +  
\alpha_{t_1} \ln\left(\frac{(e-1)^2}{e}\right) + 
\alpha_{t_2} \ln\left(\frac{e}{e-1}\right) + \alpha_{t_1}\Theta(1/n)\\
&\geq (\alpha_{t_0} + \alpha_{t_2}) \ln\left(\frac{e}{e-1}\right) + 2\alpha_{t_1} \ln\left(\frac{e-1}{\sqrt{e}}\right) + \alpha_{t_1}\Theta(1/n)\\
&\geq \alpha_{n}\ln\left(\frac{e}{e-1}\right) +\alpha_{n}\ln\left(\frac{e-1}{\sqrt{e}}\right) + \alpha_{n}\Theta(1/n)= \alpha_n/2 + \alpha_{n}\Theta(1/n)
\end{align*}
Multiplying both sides by $\lfloor n/e\rfloor/n$, we obtain that $OPT/\E[ALG] \leq 2e + O(1/n)$. 
Recalling our remark (dummy agents) from the beginning of the section, we can drop the lower-order term and consider only the limit of the approximation ratio for $n\rightarrow \infty$.
\end{proof}

\subsection{Mechanism with myopic or farsighted agents}\label{sec:secretary_mech}

In the previous section, we presented a simple $2e$-approximation in the algorithmic setting, without considering the agents' incentives.
\dmedit{As before, if the valuation functions satisfy the single crossing condition, then our $2e$-approximation algorithm can be implemented truthfully with appropriate payments (which is prompt for myopic agents and tardy for farsighted agents respectively).

\begin{corollary}\label{cor:secretary_SC}
If the valuation functions satisfy the single crossing condition, then \Cref{algo:secretary_farsighted_algo} provides an EPIC mechanism by charging price $p_t$ for the selected agent $t$, where for farsighted agents
\[
p_t = \inf\{v_t(s_{-t},s'_t) | s'_t\ge 0 \text{ s.t. } v_t(s_{[t-1]},s'_t) \ge v_i(s_{[t-1]},s'_t) \text{ for all } i>t \},
\]
and for myopic agents
\[
p_t = \inf\{v_t(s_{[t-1]},s'_t) | s'_t\ge 0 \text{ s.t. } v_t(s_{[t-1]},s'_t) \ge v_i(s_{[t-1]},s'_t) \text{ for all } i>t \}.
\]
\end{corollary}

However, without the single crossing assumption, we cannot obtain an EPIC mechanism using \Cref{algo:secretary_farsighted_algo}. In this section, we show how to achieve a $4e$-approximation with a truthful stopping rule (monotone in each agent's signal).}
The intuition behind our mechanism is quite simple: we combine the random sampling mechanism of \cite{EdenFFGK19}, which achieves a $4$-approximation in the offline setting, with the $e$-approximation stopping rule of the standard secretary problem.

\begin{algorithm}
\quad Define $t_0 = \lfloor n/2\rfloor$. 
Stop at the first time $t$ such that:
\begin{itemize}
\item $t > t_0 + \lfloor n/(2e) \rfloor$ (i.e., skip a constant fraction of agents), and
\item $v_t(\sigs_{[t_0]\cup\{t\}}) > v_i(\sigs_{[t_0]\cup\{i\}})$ for all $t_0 < i<t$.
\end{itemize}
\quad At the end, charge $t$ a price of $\inf\{v_t(\sigs_{-t}, s_t')\,|\,s_t'\geq 0\text{ s.t. }v_t(\sigs_{[t_0]}, s_t') > \max_{t_0 < i < t} v_i(\sigs_{[t_0]\cup\{i\}})\}$.
\caption{4e-approximation mechanism for the secretary model.}
\label{algo:secretary_farsighted_mechanism}
\end{algorithm}

\begin{theorem}\label{thm:secretary_farsighted_mechanism}
\Cref{algo:secretary_farsighted_mechanism} is a $4e$-approximation. {That is, the expected (myopic) value of the accepted agent is at least $\frac{1}{4e}\max_i v_i(\sigs)$.}
\end{theorem}
\begin{proof}
We define the random variable $T\in \{1, \dots, n, \infty\}$ to be the stopping time of the algorithm. In the secretary setting, $n$ agents from a set $A$ arrive in a uniformly random order $a_1, \dots, a_n$. Recall that we labeled agents according to their arrival order, that is, in the algorithm, 
$$
\forall i\in [n], \forall J\subseteq [n],\qquad
v_i(\sigs_J) := \bar v_{a_i}(\bar \sigs_{\{a_j | j\in J\}}),
$$
where $\bar v$ and $\bar \sigs$ are fixed valuation functions and signals (worst case). In particular, there exist an agent $a^\star\in A$ with the largest value $OPT = \bar v_{a^\star}(\bar\sigs)$. For convenience, we define the set function
$$
\forall X\subseteq A, \qquad f(X) := \bar v_{a^\star}(\bar s_X).
$$
Let $t_0 := \lfloor n/2\rfloor$ and $t_1 := t_0 + \lfloor n/(2e)\rfloor$. Now, observe that the stopping rule of \Cref{algo:secretary_farsighted_mechanism} does not exactly correspond to the hypothesis in \Cref{lem:secretary_stopping}. Indeed, the agent arriving at time $t$ is only compared to agents $t_0 < i < t$, so the best agent of $A_t$ depends on the order in which they arrived. This is easily fixed if we say that the first $t_0$ are here to initialize the mechanism, which only start at time $t_0+1$. More formally, for every fixed set $A_{t_0} = R$ we define
$$
\forall S\subseteq A\setminus R,\qquad \texttt{best}_{R}(S) := \text{argmax}_{a\in S} \bar v_a(\sigs_{R\cup\{a\}}).
$$
Applying \Cref{lem:secretary_stopping} to the mechanism defined by $\texttt{best}_R$ and starting at time $t_0+1$, we have
\begin{equation}\label{eq:secretary_farsighted_mechanism}
\forall t > t_1,\qquad
\Pr[T = t\,|\,A_t,A_{t_0}] = \frac{t_1-t_0}{(t-t_0)(t-t_0-1)}.
\end{equation}
Next, we write the expected welfare of the algorithm as
\begin{align*}
    \E[ALG] &\ge \sum_{t=t_1+1}^n \E[\ind[T=t]\cdot v_t(\sigs_{[t]})] &\text{\dmedit{(equality holds for myopic)}}\\
    &\geq \sum_{t=t_1+1}^n \E[\ind[T=t]\cdot v_t(\sigs_{[t_0]\cup\{t\}})\cdot\ind[a^\star \in \arrived{t}\setminus\arrived{t_0}]]&\text{(always smaller)}\\
    &\geq \sum_{t=t_1+1}^n \E[\ind[T=t]\cdot f(\arrived{t_0}\cup\{a^\star\})\cdot\ind[a^\star \in \arrived{t}\setminus\arrived{t_0}]]&\text{(stopping condition)}\\
    &\geq \sum_{t=t_1+1}^n \frac{t_1-t_0}{(t-t_0)(t-t_0-1)}\cdot\E[f(\arrived{t_0}\cup\{a^\star\})\cdot\ind[a^\star \in \arrived{t}\setminus\arrived{t_0}]]&\text{(using \Cref{{eq:secretary_farsighted_mechanism}})}\\
    &\geq \sum_{t=t_1+1}^n \frac{t_1-t_0}{(t-t_0)(t-t_0-1)}\cdot\frac{t-t_0}{n}\cdot \E[f(\arrived{t_0}\cup\{a^\star\})]& (\Pr[a^\star\in\arrived{t}\setminus\arrived{t_0}]=(t-t_0)/n)
\end{align*}
Using the subadditivity of $f$, we have that
\begin{align*}
OPT = f(A) &= \E[f(A)]&\text{(by definition)}\\
&\leq \E[f(A_{t_0}\cup\{a^\star\}) + f(A\setminus (A_{t_0}\cup\{a^\star\}))]&\text{(by subadditivity)}\\
&\leq 2\cdot \E[f(A_{t_0}\cup\{a^\star\})]&\text{(by symmetry)}
\end{align*}
Overall, we obtain that
\begin{align*}
\E[ALG] &\geq \frac{OPT}{2}\cdot \frac{(t_1-t_0)}{n}\sum_{t=t_1+1}^n \frac{1}{t-t_0-1}\geq OPT\cdot \frac{(t_1-t_0)}{2n} \ln\left(\frac{n-t_0}{t_1-t_0}\right)\geq OPT\cdot \frac{\lfloor n/(2e)\rfloor}{2n}
\end{align*}
Overall, we obtain $OPT/\E[ALG] \leq 4e + O(1/n)$.
Once again, using the remark made at the beginning of the section, we can drop the lower order term $O(1/n)$ by taking $n\to \infty$.

\end{proof}

We observe that \Cref{algo:secretary_farsighted_mechanism} is a truthful mechanism by using~\Cref{lem:truthful}.

\begin{lemma}\label{lem:sec_truthful}
    \Cref{algo:secretary_farsighted_mechanism} is an EPIC mechanism
\end{lemma}
\begin{proof}
    This follows by observing that the allocation for agent $t$ is monotone non-decreasing in her signal $s_t$. Because: first, none of the sample agents will be allocated no matter their signal, and second, and at time $t$ agent $t$'s signal is only used to to determine her own estimated value $v_t(\sigs_{[t_0]\cup\{t\}})$ and this is monotone in her signal $s_t$. Thus we can charge the appropriate payments given by \Cref{lem:truthful} to obtain EPIC mechanism \dmedit{for both the myopic and farsighted settings. Moreover, for the myopic setting the value and payments don't depend on future signals, and hence, the mechanism can be implemented with prompt payments.}
\end{proof}

\subsection{Extensions beyond our mechanisms}\label{extensionprivatevals}
We remark that the secretary setting with myopic agents is especially well-behaved. In particular, the random order gives (in expectation) an outlook on all agents' valuations, since we see a random subset of them (i.e. the yet-arrived ones), evaluated on a random subset of the signals. This continues to hold even if the valuations are \emph{not} known, but private information of the agents.
Only recently, \citet{EdenGZ22,EdenFGMM23} give a constant approximation mechanism in the offline setting with submodular valuations (approximation ratio $5.55$). A construction similar to that of \Cref{algo:secretary_farsighted_mechanism} can be combined with their result to obtain a constant approximation mechanism for the secretary problem with myopic agents and submodular valuations. However, 
obtaining similar results for more challenging settings than this (secretary, myopic) seems out of reach with standard reductions, as having private valuations gives too much strategic power to the agents interacting with the mechanism.

\section{Conclusion}
Our results consider agents with interdependent valuations in context of the celebrated secretary and prophet inequalities problems, capturing (stochastic) online versions of single-item auctions with agents that exhibit interdependent valuations.
We give the first secretary and prophet algorithms and mechanisms for this setting, achieving small constant-factor approximations to the according standard benchmarks. This resolves (up to possibly improvements in the constant) the according algorithmic and mechanism design problems for both myopic and farsighted agents with public, subadditive valuations. 
The fact that our results are all constructive and obtained by simple stopping rules especially raises hope that in future work, they can be extended to different settings.
As one direction, it is an interesting question to investigate the case of private valuations, and prove for which settings constant approximations are/are not possible (see short discussion for the secretary setting in \cref{extensionprivatevals}).

While arguably, our considered class of subadditive valuation functions poses a natural barrier for the performance of online algorithms, one should also investigate performance of such stopping rules when values satisfy other properties from the hierarchy of complement-free valuations defined by \citet{LehmannLN06} (see also the discussion for submodular valuations in \cref{app:submodular}).

Finally, extending from the simple setting, where only one online agent can be chosen, towards combinatorial problems (e.g., with cardinality/matroid constraints on the chosen subset, or for certain combinatorial auctions) would be a consequential next step.

\newpage
\bibliographystyle{plainnat}
\bibliography{EC24/biblio}

\begin{thebibliography}{45}
\providecommand{\natexlab}[1]{#1}
\providecommand{\url}[1]{\texttt{#1}}
\expandafter\ifx\csname urlstyle\endcsname\relax
  \providecommand{\doi}[1]{doi: #1}\else
  \providecommand{\doi}{doi: \begingroup \urlstyle{rm}\Url}\fi

\bibitem[Alaei et~al.(2012)Alaei, Hajiaghayi, and Liaghat]{AlaeiHL12}
Saeed Alaei, MohammadTaghi Hajiaghayi, and Vahid Liaghat.
\newblock Online prophet-inequality matching with applications to ad
  allocation.
\newblock 06 2012.
\newblock \doi{10.1145/2229012.2229018}.

\bibitem[Amer and Talgam-Cohen(2021)]{AmerTC21}
Ameer Amer and Inbal Talgam-Cohen.
\newblock {Auctions with Interdependence and SOS: Improved Approximation}.
\newblock In \emph{14th International Symposium on Algorithmic Game Theory
  (SAGT)}, 2021.

\bibitem[Ausubel et~al.(1999)]{ausubel1999generalized}
Lawrence~M Ausubel et~al.
\newblock A generalized vickrey auction.
\newblock \emph{Econometrica}, 1999.

\bibitem[Babaioff et~al.(2007)Babaioff, Immorlica, and Kleinberg]{BabaioffIK07}
Moshe Babaioff, Nicole Immorlica, and Robert Kleinberg.
\newblock Matroids, secretary problems, and online mechanisms.
\newblock pages 434--443, 01 2007.

\bibitem[Birmpas et~al.(2023)Birmpas, Ezra, Leonardi, and Russo]{BirmpasELR23}
Georgios Birmpas, Tomer Ezra, Stefano Leonardi, and Matteo Russo.
\newblock Fair division with interdependent values.
\newblock \emph{CoRR}, abs/2305.14096, 2023.
\newblock \doi{10.48550/ARXIV.2305.14096}.
\newblock URL \url{https://doi.org/10.48550/arXiv.2305.14096}.

\bibitem[Brunel and Krengel(1979)]{brunel1979parier}
Antoine Brunel and Ulrich Krengel.
\newblock Parier avec un proph{\`e}te dans le cas d'un processus sous-additif.
\newblock 1979.
\newblock URL \url{https://gallica.bnf.fr/ark:/12148/bpt6k98138170/f73.item}.

\bibitem[Chakraborty et~al.(2010)Chakraborty, Citanna, and
  Ostrovsky]{chakraborty2010two}
Archishman Chakraborty, Alessandro Citanna, and Michael Ostrovsky.
\newblock Two-sided matching with interdependent values.
\newblock \emph{Journal of Economic Theory}, 145\penalty0 (1):\penalty0
  85--105, 2010.

\bibitem[Chawla et~al.(2010)Chawla, Hartline, Malec, and Sivan]{ChawlaSH10}
Shuchi Chawla, Jason Hartline, David Malec, and Balasubramanian Sivan.
\newblock Sequential posted pricing and multi-parameter mechanism design.
\newblock \emph{Northwestern University, Center for Mathematical Studies in
  Economics and Management Science, Discussion Papers}, 06 2010.
\newblock \doi{10.1145/1807406.1807428}.

\bibitem[Chawla et~al.(2014)Chawla, Fu, and Karlin]{ChawlaFK14}
Shuchi Chawla, Hu~Fu, and Anna Karlin.
\newblock Approximate revenue maximization in interdependent value settings.
\newblock In \emph{15th ACM Conference on Economics and Computation (EC)},
  2014.

\bibitem[Che et~al.(2015)Che, Kim, and Kojima]{CKK15}
Yeon-Koo Che, Jinwoo Kim, and Fuhito Kojima.
\newblock Efficient assignment with interdependent values.
\newblock \emph{Journal of Economic Theory}, 158:\penalty0 54--86, 2015.

\bibitem[Chen et~al.(2022)Chen, Eden, and Wang]{ChenEW}
Yiling Chen, Alon Eden, and Juntao Wang.
\newblock {Cursed yet Satisfied Agents}.
\newblock In \emph{13th Innovations in Theoretical Computer Science Conference
  (ITCS 2022)}, 2022.

\bibitem[Cohen et~al.(2023)Cohen, Feldman, Mohan, and
  Talgam{-}Cohen]{CohenFMT23}
Avi Cohen, Michal Feldman, Divyarthi Mohan, and Inbal Talgam{-}Cohen.
\newblock Interdependent public projects.
\newblock In \emph{34th {ACM-SIAM} Symposium on Discrete Algorithms (SODA)},
  2023.

\bibitem[Correa and Cristi(2023)]{CorreaC23}
Jose Correa and Andrés Cristi.
\newblock A constant factor prophet inequality for online combinatorial
  auctions.
\newblock pages 686--697, 06 2023.
\newblock \doi{10.1145/3564246.3585151}.

\bibitem[Dasgupta and Maskin(2000)]{DM00}
Partha Dasgupta and Eric Maskin.
\newblock Efficient auctions.
\newblock \emph{The Quarterly Journal of Economics}, 2000.

\bibitem[Dynkin(1963)]{Dynkin63}
E.~Dynkin.
\newblock The optimum choice of the instant for stopping a markov process.
\newblock \emph{Soviet Mathematics. Doklady}, 4, 01 1963.

\bibitem[Dütting et~al.(2020)Dütting, Feldman, Kesselheim, and
  Lucier]{DuttingFKL20}
Paul Dütting, Michal Feldman, Thomas Kesselheim, and Brendan Lucier.
\newblock Prophet inequalities made easy: Stochastic optimization by pricing
  nonstochastic inputs.
\newblock \emph{SIAM Journal on Computing}, 49:\penalty0 540--582, 04 2020.
\newblock \doi{10.1137/20M1323850}.

\bibitem[Eden et~al.(2018)Eden, Feldman, Fiat, and Goldner]{EdenFFG18}
Alon Eden, Michal Feldman, Amos Fiat, and Kira Goldner.
\newblock Interdependent values without single-crossing.
\newblock In \emph{19th ACM Conference on Economics and Computation (EC)},
  2018.

\bibitem[Eden et~al.(2019)Eden, Feldman, Fiat, Goldner, and Karlin]{EdenFFGK19}
Alon Eden, Michal Feldman, Amos Fiat, Kira Goldner, and Anna~R. Karlin.
\newblock Combinatorial auctions with interdependent valuations: Sos to the
  rescue.
\newblock In \emph{20th ACM Conference on Economics and Computation (EC)},
  2019.

\bibitem[Eden et~al.(2021)Eden, Feldman, Talgam{-}Cohen, and Zviran]{EdenFTZ21}
Alon Eden, Michal Feldman, Inbal Talgam{-}Cohen, and Ori Zviran.
\newblock Poa of simple auctions with interdependent values.
\newblock In \emph{35th {AAAI} Conference on Artificial Intelligence (AAAI)},
  2021.

\bibitem[Eden et~al.(2022)Eden, Goldner, and Zheng]{EdenGZ22}
Alon Eden, Kira Goldner, and Shuran Zheng.
\newblock Private interdependent valuations.
\newblock In \emph{33rd {ACM-SIAM} Symposium on Discrete Algorithms (SODA)},
  2022.

\bibitem[Eden et~al.(2023)Eden, Feldman, Goldner, Mauras, and
  Mohan]{EdenFGMM23}
Alon Eden, Michal Feldman, Kira Goldner, Simon Mauras, and Divyarthi Mohan.
\newblock Constant approximation for private interdependent valuations.
\newblock In \emph{64th {IEEE} Annual Symposium on Foundations of Computer
  Science, {FOCS} 2023, Santa Cruz, CA, USA, November 6-9, 2023}, pages
  148--163. {IEEE}, 2023.
\newblock \doi{10.1109/FOCS57990.2023.00018}.

\bibitem[for the Prize in Economic Sciences in Memory~of
  Alfred~Nobel(2020)]{nobel2021considerations}
The~Committee for the Prize in Economic Sciences in Memory~of Alfred~Nobel.
\newblock Scientifc background on the sveriges riksbank prize in economic
  sciences, 2020.
\newblock URL
  \url{https://www.nobelprize.org/uploads/2020/09/advanced-economicsciencesprize2020.pdf}.

\bibitem[Gkatzelis et~al.(2021)Gkatzelis, Patel, Pountourakis, and
  Schoepflin]{gkatzelis2021prior}
Vasilis Gkatzelis, Rishi Patel, Emmanouil Pountourakis, and Daniel Schoepflin.
\newblock Prior-free clock auctions for bidders with interdependent values.
\newblock In \emph{14th International Symposium on Algorithmic Game Theory
  (SAGT)}, 2021.

\bibitem[Hajiaghayi et~al.(2007)Hajiaghayi, Kleinberg, and
  Sandholm]{HajiaghayiKS07}
Mohammad Hajiaghayi, Robert Kleinberg, and Tuomas Sandholm.
\newblock Automated online mechanism design and prophet inequalities.
\newblock volume~1, pages 58--65, 01 2007.

\bibitem[Immorlica et~al.(2023)Immorlica, Singla, and Waggoner]{ImmorlicaSW23}
Nicole Immorlica, Sahil Singla, and Bo~Waggoner.
\newblock Prophet inequalities with linear correlations and augmentations.
\newblock \emph{ACM Transactions on Economics and Computation}, 11, 09 2023.
\newblock \doi{10.1145/3623273}.

\bibitem[Ito and Parkes(2006)]{ItoP06}
Takayuki Ito and David~C. Parkes.
\newblock Instantiating the contingent bids model of truthful interdependent
  value auctions.
\newblock In \emph{Proceedings of the Fifth International Joint Conference on
  Autonomous Agents and Multiagent Systems}, page 1151–1158, 2006.

\bibitem[Jehiel and Moldovanu(2001)]{JehieM01}
Philippe Jehiel and Benny Moldovanu.
\newblock Efficient design with interdependent valuations.
\newblock \emph{Econometrica}, 69\penalty0 (5):\penalty0 1237--1259, 2001.

\bibitem[Jehiel et~al.(2006)Jehiel, ter Vehn, Moldovanu, and Zame]{JehielMMZ06}
Philippe Jehiel, Moritz~Meyer ter Vehn, Benny Moldovanu, and William~R. Zame.
\newblock The limits of ex post implementation.
\newblock \emph{Econometrica}, 74\penalty0 (3):\penalty0 585--610, 2006.
\newblock ISSN 00129682, 14680262.

\bibitem[Kesselheim et~al.(2013)Kesselheim, Radke, Abels, and
  Vöcking]{AbelsKRV13}
Thomas Kesselheim, Klaus Radke, Andreas Abels, and Berthold Vöcking.
\newblock An optimal online algorithm for weighted bipartite matching and
  extensions to combinatorial auctions.
\newblock 09 2013.
\newblock ISBN 978-3-642-40449-8.
\newblock \doi{10.1007/978-3-642-40450-4_50}.

\bibitem[Kleinberg and Weinberg(2012)]{KleinbergW12}
Robert Kleinberg and S.~Weinberg.
\newblock Matroid prophet inequalities.
\newblock \emph{Proceedings of the Annual ACM Symposium on Theory of
  Computing}, 01 2012.
\newblock \doi{10.1145/2213977.2213991}.

\bibitem[Klemperer(1998)]{Klemperer98}
Paul Klemperer.
\newblock Auctions with almost common values: The 'wallet game' and its
  applications.
\newblock \emph{European Economic Review}, 42\penalty0 (3-5):\penalty0
  757--769, 1998.
\newblock URL
  \url{https://EconPapers.repec.org/RePEc:eee:eecrev:v:42:y:1998:i:3-5:p:757-769}.

\bibitem[Krengel and Sucheston(1978)]{KrengelS78}
Ulrich Krengel and L.~Sucheston.
\newblock On semiamarts, amarts, and processes with finite value.
\newblock \emph{Probability on Banach Spaces}, pages 197--266, 01 1978.

\bibitem[Krengel and Sucheston(1977)]{KrengelS77}
Ulrich Krengel and Louis Sucheston.
\newblock Semiamarts and finite values.
\newblock \emph{Bulletin of The American Mathematical Society - BULL AMER MATH
  SOC}, 83, 10 1977.
\newblock \doi{10.1090/S0002-9904-1977-14378-4}.

\bibitem[Lehmann et~al.(2006)Lehmann, Lehmann, and Nisan]{LehmannLN06}
Benny Lehmann, Daniel Lehmann, and Noam Nisan.
\newblock Combinatorial auctions with decreasing marginal utilities.
\newblock \emph{Games Econ. Behav.}, 55\penalty0 (2):\penalty0 270--296, 2006.

\bibitem[Li(2013)]{Li13}
Yunan Li.
\newblock Approximation in mechanism design with interdependent values.
\newblock In \emph{Proceedings of the fourteenth {ACM} Conference on Electronic
  Commerce, {EC} 2013, Philadelphia, PA, USA, June 16-20, 2013}, pages
  675--676. {ACM}, 2013.
\newblock \doi{10.1145/2492002.2482580}.
\newblock URL \url{https://doi.org/10.1145/2492002.2482580}.

\bibitem[Lu et~al.(2022)Lu, Sun, and Zhou]{LuSZ22}
Pinyan Lu, Enze Sun, and Chenghan Zhou.
\newblock Better approximation for interdependent {SOS} valuations.
\newblock In \emph{{WINE}}, volume 13778 of \emph{Lecture Notes in Computer
  Science}, pages 219--234. Springer, 2022.

\bibitem[Maskin(1992)]{maskin1992}
Eric Maskin.
\newblock Auctions and privatization.
\newblock \emph{Privatization, H. Siebert, ed. (Institut fur Weltwirtschaften
  der Universit\"{a}t Kiel: 1992)}, 1992.

\bibitem[Milgrom and Weber(1982)]{MilgromWeber82}
Paul~R Milgrom and Robert~J Weber.
\newblock A theory of auctions and competitive bidding.
\newblock \emph{Econometrica}, 1982.

\bibitem[Myerson(1981)]{myerson1981optimal}
Roger~B Myerson.
\newblock Optimal auction design.
\newblock \emph{Mathematics of operations research}, 6\penalty0 (1):\penalty0
  58--73, 1981.

\bibitem[Reiffenh{\"{a}}user(2019)]{Reiffenhauser19}
Rebecca Reiffenh{\"{a}}user.
\newblock An optimal truthful mechanism for the online weighted bipartite
  matching problem.
\newblock In Timothy~M. Chan, editor, \emph{Proceedings of the Thirtieth Annual
  {ACM-SIAM} Symposium on Discrete Algorithms, {SODA} 2019, San Diego,
  California, USA, January 6-9, 2019}, pages 1982--1993. {SIAM}, 2019.
\newblock \doi{10.1137/1.9781611975482.120}.
\newblock URL \url{https://doi.org/10.1137/1.9781611975482.120}.

\bibitem[Roughgarden and Talgam-Cohen(2016)]{RoughgardenTC16}
Tim Roughgarden and Inbal Talgam-Cohen.
\newblock Optimal and robust mechanism design with interdependent values.
\newblock \emph{ACM Trans. Econ. Comput.}, 2016.

\bibitem[Samuel-Cahn(1984)]{SamuelCahn84}
Ester Samuel-Cahn.
\newblock Comparison of threshold stop rules and maximum for independent
  nonnegative random variables.
\newblock \emph{The Annals of Probability}, 12, 11 1984.
\newblock \doi{10.1214/aop/1176993150}.

\bibitem[Samuel-Cahn(1991)]{SamuelCahn91}
Ester Samuel-Cahn.
\newblock Prophet inequalities for bounded negatively dependent random
  variables.
\newblock \emph{Statistics \& Probability Letters}, 12:\penalty0 213--216, 09
  1991.
\newblock \doi{10.1016/0167-7152(91)90080-B}.

\bibitem[Wilson(1969)]{wilson1969communications}
Robert~B Wilson.
\newblock Communications to the editor—competitive bidding with disparate
  information.
\newblock \emph{Management science}, 1969.

\bibitem[Yao(1977)]{Yao77}
Andrew~Chi{-}Chih Yao.
\newblock Probabilistic computations: Toward a unified measure of complexity
  (extended abstract).
\newblock In \emph{{FOCS}}, pages 222--227. {IEEE} Computer Society, 1977.

\end{thebibliography}
\newpage

\appendix

\section{Improved bounds for submodular valuations}\label{app:submodular}

In this section we present an online algorithm that obtains an improved approximation ratio of $4$ for the secretary model when the valuations are submodular over signals (for both the myopic and farsighted settings). 

\begin{definition}[Submodular over signals]
    We say a valuation function $v(\cdot)$ is submodular over signals, if for any $i \in [n]$ and signal profiles $\sigs\ge \sigs'$ we have
    \[
    v(s_i,\sigs_{-i}) - v(s'_i,\sigs_{-i}) \le v(s_i,\sigs'_{-i}) - v(s'_i,\sigs'_{-i})
    \]
\end{definition}

The following lemma is a generalization of the Key lemma from \cite{EdenFFGK19}, which was proved in \cite{LuSZ22}.
\begin{lemma}\label{lem:submodular-bound}
    For any monotone submodular valuation function $v$, and for any random subset $A\subseteq [n]$ such that $A$ is drawn uniformly among subsets of size $k$, we have
    \[
    \E_{A}[v(\sigs_A)] \ge \frac{k}{n}\cdot v(\sigs)
    \]
\end{lemma}

The algorithm is a slight modification of \Cref{algo:secretary_farsighted_algo}. In particular, the sampling phase involves $n/2$ agents instead of $n/e$.

\begin{algorithm}
\quad At step $t$, when agent $t$ arrives, stop if:
\begin{itemize}
\item $t > n/2$ (i.e., skip a constant fraction of agents), and
\item $v_t(\sigs_{[t]}) > v_i(\sigs_{[t]})$ for all $i<t$.
\end{itemize}
\caption{4-approximation algorithm under submodular valuations.}
\label{algo:secretary_sos}
\end{algorithm}

We are now ready to prove the main results of this section. In fact we prove a stronger statement that, in expectation, the myopic value of the accepted agent is a $4$-approximation to the farsighted benchmark. This immediately implies a $4$-approximation for both the myopic and farsighted settings.

\begin{theorem}\label{thm:sos-4approx}
    \Cref{algo:secretary_sos} is a $4$-approximation under submodular valuations. That is, the expected (myopic) values of the accepted agent is at least $\frac{1}{4}\max_i v_i(\sigs)$.
\end{theorem}

{
\begin{proof}
We define the random variable $T\in \{1, \dots, n, \infty\}$ to be the stopping time of the algorithm. In the secretary setting, $n$ agents from a set $A$ arrive in a uniformly random order $a_1, \dots, a_n$. Recall that we labeled agents according to their arrival order, that is, in the algorithm, 
$$
\forall i\in [n], \forall J\subseteq [n],\qquad
v_i(\sigs_J) := \bar v_{a_i}(\bar \sigs_{\{a_j | j\in J\}}),
$$
where $\bar v$ and $\bar \sigs$ are the original, fixed valuation functions and signals (determined adversarially) before applying the random ordering. In particular, there exists an agent $a^\star\in A$ with the largest value $OPT = \bar v_{a^\star}(\bar\sigs)$. For convenience, we define the set function
$$
\forall X\subseteq A, \qquad f(X) := \bar v_{a^\star}(\bar s_X),
$$
that is, $f(X)$ denotes the estimated value of $a^\star$ only using the signals of $X\subseteq A$.

Next, we define the (random) set $\arrived{t} := \{a_1, \dots, a_t\}$ of agents who have arrived at time $t$. Observe that the stopping rule of \Cref{algo:secretary_sos} corresponds to \Cref{lem:secretary_stopping} with $k = \lfloor n/2\rfloor$ and 
$$
\forall S\subseteq A,\qquad \texttt{best}(S) := \text{argmax}_{a\in S} \bar v_a(\bar \sigs_S).
$$
Using \cref{lem:secretary_stopping}, the event where the algorithm stops at time $T=t$ is independent of $\arrived{t}$, and has probability equal to
\begin{equation}\label{eq:secretary_farsighted_algo_sos}
\forall t > n/2,\qquad
\Pr[T = t\,|\,\arrived{t}] = \frac{\lfloor n/2\rfloor}{t(t-1)}
\end{equation}
We write the expected welfare obtained by the algorithm as
\begin{align*}
    \E[ALG] &\ge \sum_{t=\lceil n/2\rceil}^n \E[\ind[T=t]\cdot v_t(\sigs_{[t]})]&\text{\dmedit{(equality holds for myopic)}}\\
    &\geq \sum_{t=\lceil n/2\rceil}^n \E[\ind[T=t]\cdot v_t(\sigs_{[t]})\cdot\ind[a^\star \in \arrived{t}]]&\text{(always smaller)}\\
    &\geq \sum_{t=\lceil n/2\rceil}^n \E[\ind[T=t]\cdot f(\arrived{t})\cdot\ind[a^\star \in \arrived{t}]]&\text{(by the stopping condition)}\\
    &= \sum_{t=\lceil n/2\rceil}^n \frac{\lfloor n/2\rfloor}{t(t-1)}\cdot\E[f(\arrived{t})\cdot\ind[a^\star \in \arrived{t}]]&\text{(using \Cref{eq:secretary_farsighted_algo_sos})}\\
    &= \sum_{t=\lceil n/2\rceil}^n \frac{\lfloor n/2\rfloor}{t(t-1)}\cdot\frac{t}{n}\cdot \E[f(\arrived{t})\,|\,a^\star \in \arrived{t}]&\text{($a^\star\in\arrived{t}$ with probability $t/n$)}
\end{align*}
Next we define
$$\forall t \in [n], \qquad \alpha_t := \E[f(\arrived{t})\,|\,a^\star \in \arrived{t}].$$
which gives the inequality
$$
\E[ALG] \geq \frac{\lfloor n/2\rfloor}{n} \sum_{t=\lceil n/2\rceil}^n \frac{\alpha_t}{t-1}.
$$
Alternatively, $\alpha_t$ it is the expected value of $f(X\cup\{a^\star\})$, given a random subset $X\subseteq A\setminus\{a^\star\}$ of size $|X| = t-1$. In particular, using \Cref{lem:submodular-bound} we have that 
$\alpha_t \geq \frac{t-1}{n-1} \cdot OPT$. Therefore
$$
\E[ALG] \geq \frac{\lfloor n/2\rfloor}{n} \cdot\frac{n-\lceil n/2\rceil +1}{n-1} \cdot OPT \geq \frac{OPT}{4},
$$
which concludes the proof.
\end{proof}
}

\end{document}